\documentclass[10pt,journal]{IEEEtran}
\usepackage{amsmath,amssymb,amsthm}
\usepackage{graphicx}
\usepackage{booktabs}
\usepackage[hidelinks]{hyperref}
\usepackage{url}

\newtheorem{lemma}{Lemma}
\newtheorem{theorem}{Theorem}
\newtheorem{proposition}{Proposition}

\DeclareMathOperator{\Var}{Var}
\DeclareMathOperator{\Cov}{Cov}
\newcommand{\E}{\mathbb{E}}
\newcommand{\R}{\mathbb{R}}
\newcommand{\N}{\mathbb{N}}
\newcommand{\norm}[1]{\left\lVert#1\right\rVert}

\begin{document}

\title{Deterministic Johnson--Lindenstrauss Projections from Pisot $\beta$-Transformations for Zero-Knowledge Private Routing}

\author{I. Dey~\IEEEmembership{Senior Member,~IEEE} and I. Cherkaoui%
\thanks{I. Cherkaoui, and I. Dey are with South East Technological University, Waterford, Ireland (Email: ilias.cherkaoui@setu.ie, indrakshi.dey@setu.ie). This work is supported in part by HEARG TU RISE Project ``AIQ-Shield" and the HORIZON ECCC Project ``Q-FENCE" under Grant Number 101225708. \emph{This paper is under review at IEEE Signal Processing Letters.}}}

\markboth{IEEE Signal Processing Letters (Submitted)}{}
\maketitle

\begin{abstract}
Zero-knowledge (ZK) proofs certify that a message belongs to an allowed semantic class without revealing the message, but the certificate compares a high-dimensional embedding against class centroids, so its cost grows with the embedding dimension $d$. A Johnson--Lindenstrauss (JL) projection lowers $d$ to $m\ll d$ while preserving pairwise distances, yet a random JL matrix must be committed and its sampling proved inside the circuit, which is costly and a leakage risk. We construct a public deterministic projection from the standardized orbit of a Pisot $\beta$-transformation, analyzed through the spectral gap of the $\beta$-map, the geometric decay of its correlations, rather than equidistribution. We prove that the induced squared-norm estimator is unbiased up to a term decaying geometrically with a sampling gap, and that its variance is $V_0/m$ with a constant $V_0$ that is dimension-free in experiment and, under one stated concentration hypothesis, in theory. A single public seed preserving all pairwise centroid distances therefore exists and is found by search. Against six standard projections, including the chaotic-sequence matrix of Yu \emph{et al.}, the construction matches statistical quality to within measurement noise, and it is the only one simultaneously free of in-circuit randomness and exactly reproducible in a fixed finite field at a per-step cost $\log_2\beta$ rather than $2^{k}$.
\end{abstract}

\begin{IEEEkeywords}
Johnson--Lindenstrauss, random projection, dimensionality reduction, zero-knowledge proofs, $\beta$-expansion, Pisot numbers, deterministic sensing matrices.
\end{IEEEkeywords}

\section{Introduction}
\IEEEPARstart{A}{utonomous} AI agents increasingly cooperate across organizational and trust boundaries. A clinical triage agent forwards a case to a specialist, a compliance agent screens a payment instruction, a router dispatches a prompt to one of several expert models. The routed message is private, yet the routing decision must be verifiable, so that an auditor or counterparty can confirm the message reached only a permitted destination. A zero-knowledge (ZK) proof serves both needs at once: it certifies that a statement is true while revealing nothing beyond its truth \cite{touheed,dsperse,tong}. The statement here is that a semantic embedding $v\in\R^d$ of the message lies in an allowed class, decided by comparing $v$ against a fixed public set of centroids $C=\{c_1,\dots,c_N\}\subset\R^d$.

The cost of this comparison limits deployment. A ZK proof compiles every arithmetic operation into a constraint, and prover time, memory, and proof size all grow with the constraint count. The nearest-centroid test uses $O(Nd)$ multiply-add constraints, and modern embeddings have $d=768$ to $4096$, so this term dominates the proof. Lowering the working dimension to $m\ll d$ lowers prover cost in direct proportion, which in a live system is the difference between a certificate produced in seconds and one taking minutes and gigabytes.

The classical remedy is a Johnson--Lindenstrauss (JL) projection \cite{jl,achlioptas}. It maps the data to dimension $m$ and preserves every pairwise distance to a factor $(1\pm\varepsilon)$, because the projection forms many weighted averages of the coordinates and averages concentrate. Such projections underpin compressed sensing and fast learning \cite{boutsidis}, and structured variants lower their cost \cite{ailon,do}. The obstacle specific to ZK is that the JL matrix is random. It is drawn after the data is fixed, so the circuit must commit to it and prove correct sampling, which restores the cost the projection was meant to remove. A deployable projection must instead be public and fixed, identical across provers and auditors, and native to the circuit's finite field, while still preserving distances provably. Deterministic matrices are usually studied through the restricted isometry property. Blanchard \emph{et al.} \cite{blanchard} show these constants need not decay for deterministic matrices, high-compression Bernoulli matrices are analyzed in \cite{lu}, and, closest to us, chaotic-sequence matrices are validated empirically \cite{yu}. We derive a dimension-free variance constant from the transfer-operator spectral gap rather than measuring it, and we use the Pisot property for exact finite-field reproducibility, which generic chaotic maps lack.

We obtain a matrix from a Pisot $\beta$-transformation, a stretch-and-fold chaotic map on the unit interval. Four contributions follow. First, we prove distance preservation from the spectral gap of the $\beta$-map (Lemmas~\ref{lem:bias} and~\ref{lem:var}), which yields a dimension-free bound that an equidistribution argument cannot. Second, we show a single public seed preserves all $\binom{N}{2}$ centroid distances and is found by an efficient search (Theorem~\ref{thm:seed}), and we give a closed-form certifiable robustness radius (Proposition~\ref{prop:rob}). Third, we identify why the Pisot property, and not generic chaos, is required: the orbit is exactly reproducible in a fixed finite field at a per-step cost equal to the map's metric entropy, whereas a generic chaotic map costs $O(2^{k})$ bits (Section~\ref{sec:repro}). Fourth, we compare against six standard projections across five studies and make the physical meaning of each result explicit (Section~\ref{sec:exp}).

\section{Construction}
Let $\beta>1$ be a Pisot number, an algebraic integer all of whose Galois conjugates lie strictly inside the unit disk; the golden ratio $1.618$ and the plastic number $1.325$ are the two smallest examples. The $\beta$-transformation $T_\beta x=\beta x-\lfloor\beta x\rfloor$ stretches the interval by $\beta$ and folds it back, the prototype of deterministic mixing, and has a unique invariant Parry measure $\mu$ \cite{parry}. Because $\beta$ is Pisot, the $\beta$-expansion of $1$ is eventually periodic, the map has a finite Markov partition, and its transfer operator has a \emph{spectral gap} \cite{hofbauer}: correlations between an observable now and $k$ steps later decay geometrically. For the centered, normalized coordinate observable $\psi(x)=(x-\bar\mu)/s_\mu$,
\begin{equation}\label{eq:sg}
r_k:=\int \psi\,(\psi\circ T_\beta^{\,k})\,d\mu,\qquad |r_k|\le C\,\theta^{\,k},\quad r_0=1 .
\end{equation}
In signal terms the orbit is a bounded, unit-variance sequence whose autocorrelation collapses after a short lag, so it behaves like weakly colored noise, nearly white beyond a correlation time set by $\theta$. We fix a sampling gap $g\in\N$, retaining every $g$-th orbit point so that stored samples fall beyond that correlation time, and a public seed $x_0$. Writing $z_p=\psi(T_\beta^{\,pg}x_0)$, we fill $A=A(x_0)\in\R^{m\times d}$ row by row, $A_{ij}=z_{(i-1)d+j}$, and set the embedding $\Phi(v)=A v/\sqrt m$. For a unit vector $u$, the squared-norm estimator
\begin{equation}\label{eq:S}
S(u)=\tfrac1m\norm{Au}_2^2=\tfrac1m\textstyle\sum_{i=1}^m\langle a_i,u\rangle^2
\end{equation}
is the projected energy of $u$, and it should track the true energy $\norm{u}_2^2$ when lengths are preserved. The Pisot property enters only through reproducibility (Section~\ref{sec:repro}); the statistics below use only \eqref{eq:sg}.

\section{Theory}
Expectations $\E$ and variances $\Var$ are over the public seed $x_0\sim\mu$, and $C$ is fixed and known at circuit-design time. That the centroids are known in advance is what allows one fixed matrix to serve: no adversary can later place points in a direction the matrix handles poorly.

\begin{lemma}[Bias is dimension-free]\label{lem:bias}
For every unit vector $u$, $\big|\E S(u)-\norm{u}_2^2\big|\le 2C\theta^{g}/(1-\theta^{g})=:b(g).$
\end{lemma}
\begin{proof}
By stationarity $\E S(u)=\sum_{j,j'}u_ju_{j'}r_{g|j-j'|}$. The diagonal gives $\norm{u}_2^2$. For each lag $k\ge1$, $\sum_j|u_ju_{j+k}|\le\norm{u}_2^2=1$ by Cauchy--Schwarz, so the off-diagonal mass is at most $2\sum_{k\ge1}C\theta^{gk}$, which has no $d$-dependence.
\end{proof}
The bias is a systematic offset in the measured energy, caused only by residual coloring of the orbit. Because the map decorrelates geometrically the offset does not accumulate with dimension, and sampling below the correlation time drives it to zero: $g\ge\log(4C/\varepsilon)/\log(1/\theta)$ gives $b(g)\le\varepsilon/2$.

\begin{lemma}[Variance falls as $1/m$]\label{lem:var}
$\Var(S(u))\le V_0/m$. The elementary sup-norm estimate gives $V_0=O(B^4d^2)$ with $B=\norm{\psi}_\infty$, but this is loose. If the fourth-order correlations of $T_\beta$ are summable, then $V_0=O(1)$, uniformly in $d$.
\end{lemma}
\begin{proof}[Proof sketch]
Write $\Var(S)=m^{-2}\sum_{i,i'}\Cov(\langle a_i,u\rangle^2,\langle a_{i'},u\rangle^2)$. Diagonal terms are bounded by $B^4d^2$; rows $i\ne i'$ are separated by orbit lag $\ge gd|i-i'|$, so \eqref{eq:sg} bounds their covariance by $C(B^2d)^2\theta^{gd|i-i'|}$, a geometric tail. The $d^2$ enters only through the crude diagonal bound; replacing it by the fourth moment of the weakly dependent sum $\langle a_i,u\rangle=\sum_j a_{ij}u_j$ removes the $d$ factor whenever the joint fourth-order cumulants are summable.
\end{proof}
The variance is the measurement noise of the estimator, falling as $1/m$ by averaging $m$ nearly independent measurements. The decisive content is that $V_0$ does not grow with $d$: a projected coordinate is a sum of $d$ weakly dependent terms, hence approximately Gaussian for every $d$, and the crude $d^2$ bound is pessimistic precisely because it ignores this Gaussianization. A fixed budget $m$, independent of $d$, therefore measures a length accurately.

\begin{theorem}[A good public seed exists and is found by search]\label{thm:seed}
Fix $\varepsilon,\delta\in(0,1)$ and $g$ with $b(g)\le\varepsilon/2$.
(a) If $m\ge 2V_0N^2/(\delta\varepsilon^2)$, then the seeds $x_0$ giving $(1\pm\varepsilon)$ distortion on all $\binom N2$ centroid pairs have measure at least $1-\delta$; a candidate is checked in $O(N^2m)$ time, so a valid public $(A,x_0)$ is found by search and published, and the circuit holds no randomness.
(b) If $S(u)$ obeys a Bernstein bound $\Pr(|S-\E S|>t)\le 2e^{-cmt^2/V_1}$, available for spectral-gap maps \cite{gouezel}, then $m=O(V_1\varepsilon^{-2}\log(N/\delta))$ suffices.
\end{theorem}
\begin{proof}
For a unit difference direction, $|S-1|\le|S-\E S|+b(g)$. With $b(g)\le\varepsilon/2$, Chebyshev and Lemma~\ref{lem:var} give per-pair failure $\le 4V_0/(m\varepsilon^2)$; a union bound over $<N^2/2$ pairs and the stated $m$ give total failure $\le\delta$. Part (b) replaces Chebyshev by the sub-exponential tail.
\end{proof}
One fixed apparatus is calibrated once by searching over seeds, much as a fixed sensor array is designed to preserve the positions of known sources. The gap between the two counts is a gap between tail estimates: bounding fluctuations by variance alone forces a margin proportional to the pair count, hence $N^2$, whereas the exponentially small tails of a mixing map suppress rare large-distortion events and cost only $\log N$. The first is unconditional, the second conditional on that concentration bound, the single open step. We do not claim validity for arbitrary inputs: Blanchard \emph{et al.} \cite{blanchard} show restricted isometry constants need not decay for deterministic matrices, so we prove distance preservation only on the fixed known set.

\begin{figure*}[tp]\centering
\includegraphics[width=0.82\textwidth]{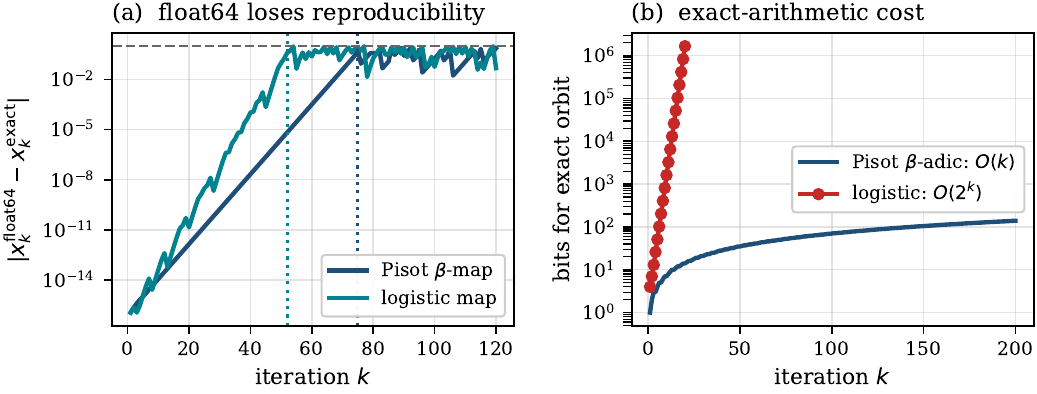}
\caption{Reproducibility. (a) The float64 orbit diverges from the exact orbit for both chaotic maps once the accumulated error reaches order one, and the onset matches the Lyapunov prediction $52\log2/\lambda$ (dotted lines), so floating point cannot define a shared matrix. (b) Cost of exact reproduction: the Pisot $\beta$-adic representation grows at the entropy rate $\log_2\beta$ per step and fits a fixed finite field, whereas the generic logistic map needs $O(2^{k})$ bits. This is the reason the Pisot property, not generic chaos, is required.}
\label{fig:repro}
\end{figure*}
\begin{figure*}[tp]\centering
\includegraphics[width=0.99\textwidth]{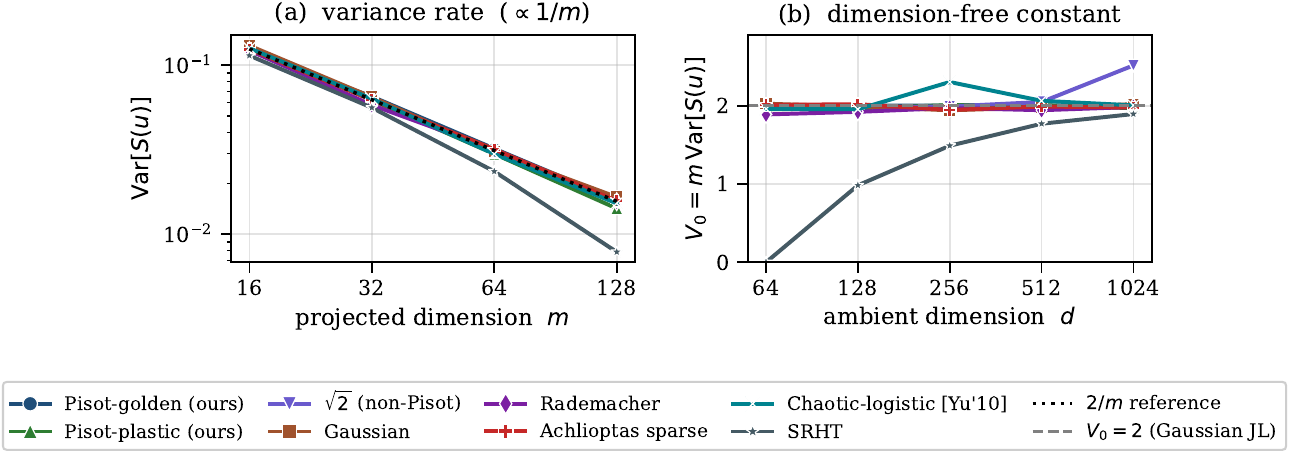}
\caption{Statistical accuracy, deterministic Pisot projection against six standard projections (shared legend). (a)~The estimator variance follows the $1/m$ law for every method and overlaps the $2/m$ reference. (b)~The variance constant $V_0=m\Var(S)$ stays flat near $2$ with no $d^2$ growth; SRHT is exact at $m=d$ and rises toward $2$.}
\label{fig:stat}
\end{figure*}

\begin{proposition}[Certifiable robustness radius]\label{prop:rob}
Let $\Phi=A/\sqrt m$ satisfy $(1\pm\varepsilon)$ distortion on $C$. Any input perturbation $\eta$ with
$\norm{\eta}_2<\rho:=\min_{i\ne j}\norm{\Phi(c_i-c_j)}_2 / (2\norm{\Phi}_{\mathrm{op}})$
cannot change the nearest projected centroid, where $\min_{i\ne j}\norm{\Phi(c_i-c_j)}_2\ge\sqrt{1-\varepsilon}\,\min_{i\ne j}\norm{c_i-c_j}_2$.
\end{proposition}
\begin{proof}
A perturbation moves the projected query by at most $\norm{\Phi}_{\mathrm{op}}\norm{\eta}_2$; if that is below half the smallest projected inter-centroid gap, the nearest centroid cannot change. The gap bound is the JL guarantee applied to each difference vector.
\end{proof}
Here $\rho$ is the noise margin of the classifier: the largest disturbance that cannot push a query past the midpoint between two classes. It is the analogue of the minimum distance of a code, and because it depends only on the public matrix it is certified in advance rather than estimated at run time.

\section{Why Pisot: Exact Reproducibility}\label{sec:repro}
Two independent provers must reproduce the same matrix $A$ bit for bit. A chaotic map amplifies error at its Lyapunov rate, which for the $\beta$-map is the stretching rate $\log\beta$, so a float64 implementation loses all agreement after about $52\log 2/\lambda$ steps. Fig.~\ref{fig:repro}(a) confirms this: the float and exact orbits diverge near $k=72$ for the $\beta$-map and $k=51$ for the logistic map, matching the Lyapunov predictions marked on the plot. Exact arithmetic is therefore required.

The cost of exact arithmetic is where the Pisot property is decisive. The digits consumed per step equal the metric entropy $\log_2\beta$, and for a Pisot $\beta$ every element of $\mathbb{Q}(\beta)$ has an eventually periodic $\beta$-expansion \cite{schmidt}, so the orbit lives in a finite set of states and $k$ steps cost $O(k\log_2\beta)$ bits, which fits a fixed finite field. A generic chaotic map has no such algebraic closure: the exact logistic orbit of a rational seed has denominators of size $3^{2^{k}}$, that is $O(2^{k})$ bits. Fig.~\ref{fig:repro}(b) shows the contrast, the logistic cost reaching $1.66\times10^{6}$ bits at $k=20$ while the $\beta$-adic cost at $k=200$ is only $139$ bits. The same quantity $\log\beta$ governs both the sensitivity to error and the exact-arithmetic budget, and only a Pisot map keeps that budget linear.

\subsubsection*{Deployment path} Offline, once $C$ is fixed, one selects a Pisot $\beta$ and searches public seeds $x_0$ until the matrix preserves every centroid distance to $\varepsilon$ (Theorem~\ref{thm:seed}); the short tuple $(\beta,x_0,m,g)$ is published as a circuit parameter that anyone can regenerate and audit. Online, the prover encodes each entry as a fixed-point $\beta$-adic integer native to the field, forms $v'=Av$ at $O(md)$ constraints, and runs the class test at $O(mN)$ rather than $O(dN)$. Because $A$ is public and fixed, no randomness is committed or proved, so the sampling and range proofs a committed random matrix needs are absent.

\section{Numerical Results}\label{sec:exp}
\begin{figure*}[tp]\centering
\includegraphics[width=0.85\textwidth]{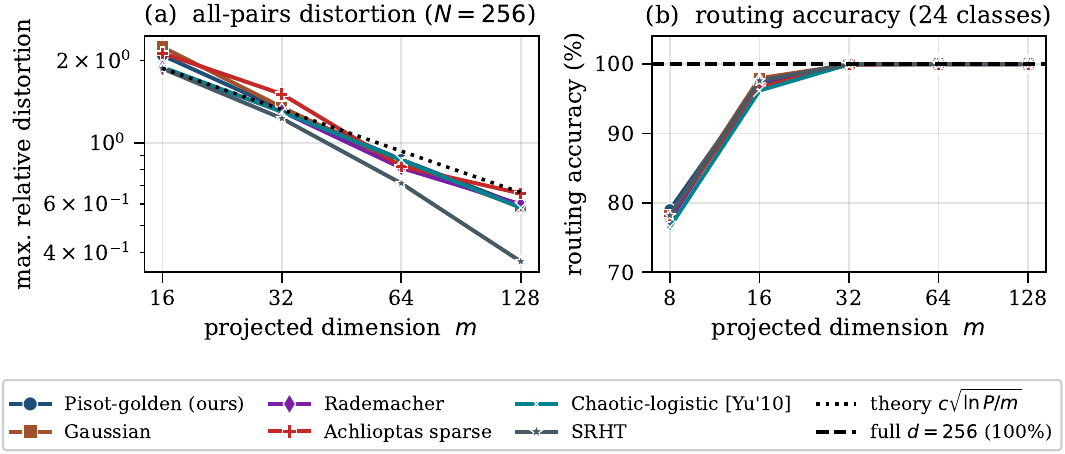}
\caption{Geometric and task fidelity (shared legend). (a)~Worst-case relative distortion over all $\binom{256}{2}$ centroid pairs tracks Gaussian JL and every baseline. (b)~Nearest-centroid routing accuracy for $24$ classes in $d=256$ recovers the full-dimension $100\%$ by $m=32$, an $8\times$ compression.}
\label{fig:geom}
\end{figure*}

We compare the two Pisot maps (golden, plastic) against six standard projections: Gaussian and Rademacher random matrices, the Achlioptas sparse ternary matrix \cite{achlioptas}, the subsampled randomized Hadamard transform (SRHT) \cite{ailon}, and the chaotic-sequence matrix of Yu \emph{et al.} \cite{yu}, the closest prior deterministic method. Every projection is produced through one interface and scored with the identical estimator, so differences reflect the matrices, not the harness, and every number below is measured, not predicted from the bounds. Statistics are taken over many independent seeds, the distribution the offline search draws from. The parameters track deployments: $d$ runs to $1024$, $N$ to $512$, and the nearest-centroid task stands in for the routing decision.

We start by plotting the Variance Rate (Fig.~\ref{fig:stat}(a)). The question is how few projections $m$, hence how cheap a proof, still measure a length reliably. Every method follows the $1/m$ law (slopes near $-1$) and overlaps the $2/m$ reference, so the deterministic matrix is as stable as a random one. Next we plot the Dimension-free constant (Fig.~\ref{fig:stat}(b)). The question is whether $m$ must grow as embeddings get larger. The constant $V_0=m\Var(S)$ stays flat near $2$ from $d=64$ to $1024$ for all methods, with no $d^2$ growth (Lemma~\ref{lem:var}); SRHT is exact at $m=d$ and rises toward $2$. A projected coordinate is Gaussianized regardless of $d$, so its variance, and the value $V_0\approx2$, are independent of $d$. The consequence is operational: the same $m$, and the same proof cost, serves a $384$- and a $4096$-dimensional embedding. The non-Pisot control $\beta=\sqrt2$ behaves identically, so the statistical quality is generic to expanding maps and the Pisot choice is spent only on reproducibility. We also represent the all-pairs distance preservation in Fig.~\ref{fig:geom}(a)). A single mispreserved pair could misroute a message, so the safety-relevant metric is the worst case over all pairs. For $N=256$ centroids ($32{,}640$ pairs) it tracks Gaussian JL and every baseline along the $c\sqrt{\ln P/m}$ trend, so the whole routing table is preserved. For downstream task (Fig.~\ref{fig:geom}(b)), on a $24$-class nearest-centroid problem in $d=256$, full-dimension accuracy is $100\%$ and every method recovers it by $m=32$, an $8\times$ compression: the routing decision is unchanged after compression, which is what a deployed router needs.

Sweeping $N$ from $64$ to $512$ jointly with $m$, the worst-case distortion surface is nearly identical for the deterministic matrix and Gaussian JL, rising slowly with $N$ and falling with $m$, matching the $\log N$ dependence of Theorem~\ref{thm:seed}: a routing table that doubles needs only a small increase in $m$. Table~\ref{tab:sota} summarizes the comparison. Statistical quality is matched across all methods; the construction is distinguished only in the ZK-relevant columns, being at once free of in-circuit randomness, exactly reproducible in a fixed finite field at the entropy rate, and backed by a spectral-gap account of the dimension-free variance constant.

\begin{table}[tb]
\caption{Comparison with standard projections. Statistical quality is matched; the last three columns are the ZK-relevant differences.}
\label{tab:sota}
\centering\footnotesize\setlength{\tabcolsep}{3.2pt}
\begin{tabular}{@{}lccccc@{}}
\toprule
Method & Var & $V_0$ flat & in-circ. & exact FF & dim-free \\
 & slope & in $d$ & random & repro. & account \\
\midrule
Pisot (ours) & $-0.98$ & yes ($2.0$) & none & yes, $O(k)$ & yes \\
Gaussian & $-1.00$ & yes ($2.0$) & $md$ reals & no & yes \\
Rademacher & $-0.98$ & yes ($2.0$) & $md$ bits & no & yes \\
Achlioptas \cite{achlioptas} & $-0.99$ & yes ($2.0$) & $\sim md$ bits & no & yes \\
SRHT \cite{ailon} & $-1.28$ & yes ($\le2$) & $d{+}$sub. & no & yes \\
Chaotic \cite{yu} & $-1.02$ & yes ($2.0$) & none & no, $O(2^k)$ & emp. \\
\bottomrule
\end{tabular}
\end{table}

\section{Conclusion}
Replacing the random JL projection in ZK private routing by the orbit of a Pisot $\beta$-transformation removes all in-circuit randomness and, in a six-way comparison, matches standard projections on variance rate, dimension-free variance constant, all-pairs distortion, and a downstream routing task, while remaining uniquely reproducible in a fixed finite field at the entropy rate $\log_2\beta$. The open theoretical step is the transfer-operator concentration inequality that would upgrade the unconditional $O(N^2)$ seed bound to the conditional $O(\log N)$ target. The practical step is an end-to-end evaluation on a concrete backend, implementing the projection and class test in R1CS or PLONK and measuring prover time, memory, and proof size against the identical circuit built on a committed random matrix that must also prove its own sampling. The constraint counts given here predict the crossover.

\end{document}